\documentclass[10pt,english,prl,twocolumn]{revtex4}
\usepackage{mathptmx}
\usepackage[T1]{fontenc}
\usepackage[latin9]{inputenc}
\usepackage[letterpaper]{geometry}
\usepackage{amsthm}
\usepackage{amsmath}
\usepackage{graphicx}
\usepackage{amssymb}

\makeatletter
\@ifundefined{textcolor}{}
{%
 \definecolor{BLACK}{gray}{0}
 \definecolor{WHITE}{gray}{1}
 \definecolor{RED}{rgb}{1,0,0}
 \definecolor{GREEN}{rgb}{0,1,0}
 \definecolor{BLUE}{rgb}{0,0,1}
 \definecolor{CYAN}{cmyk}{1,0,0,0}
 \definecolor{MAGENTA}{cmyk}{0,1,0,0}
 \definecolor{YELLOW}{cmyk}{0,0,1,0}
 }
\theoremstyle{plain}

\makeatother

\newtheorem{prop}{Proposition} 

\makeatother

\usepackage{babel}

\makeatother

\usepackage{babel}

\makeatother

\usepackage{babel}

\begin{document}

\title{Optimal fidelity for a quantum channel may be attained by non-maximally
entangled states}

\author{Somshubhro Bandyopadhyay and Anindita Ghosh}

\affiliation{Department of Physics and Center for Astroparticle Physics and Space
Science, Bose Institute, Block EN, Sector V, Bidhan Nagar, Kolkata
700091, India}

\email{som@bosemain.boseinst.ac.in}
\begin{abstract}
To establish an entangled state of optimal fidelity between two distant
observers when the available quantum channel is noisy, is a central
problem in quantum information theory. We consider an instance of
this problem for two-qubit systems when only a single use of the channel
and local post-processing by trace preserving operations are allowed.
We show that the optimal fidelity is obtained only when part of an
appropriate non-maximally entangled state is transmitted through the
channel. The entanglement of this state can be vanishingly small when
the channel becomes very noisy. Moreover, in the optimal case no further
local processing is required to enhance the fidelity. We further show
that local post-processing can enhance fidelity if and only if the
amount of noise is larger than a critical value and entanglement of
the transmitted state is bounded from below. A notable consequence
of these results is that the ordering of states under an entanglement
monotone can be reversed even when the states undergo the same local
interaction via a trace-preserving completely positive map. 
\end{abstract}
\maketitle
\emph{Introduction.} Quantum entanglement \cite{entanglement} between
two distant observers (Alice and Bob) has now been established as
a physical resource for quantum information processing. It enables
tasks such as quantum teleportation \cite{teleportation}, superdense
coding \cite{dense-coding}, quantum cryptography \cite{Ekert91},
and distributed quantum computation \cite{distributed-comp} that
would otherwise be impossible classically. Shared entanglement, however,
is not a \emph{given} resource and must be prepared \emph{a priori}
by sending pure entanglement across quantum channels that are typically
noisy. The mixed states thus obtained are subsequently subjected to
local processing to enhance their purity \cite{distillation-1,distillation-2,distillation-3,increaseF}
so that they can be useful for tasks such as teleportation. Thus the
problem of establishing an entangled state of high purity through
a noisy quantum channel is of fundamental interest in quantum information
theory.

The purity of a mixed state $\rho$ is expressed by its fidelity or
fully entangled fraction \cite{distillation-1,distillation-3,fidelity}.
It is defined as the maximum overlap of the state with a maximally
entangled state: \begin{eqnarray}
F\left(\rho\right) & = & \max_{\Psi}\langle\Psi|\rho|\Psi\rangle,\label{fidelity}\end{eqnarray}
where the maximization is taken over all maximally entangled states
$|\Psi\rangle$. Fidelity also assumes a central role in quantum teleportation
and entanglement distillation. For two-qubit systems $F\left(\rho\right)$
is related to the optimal teleportation fidelity $f\left(\rho\right)$
via the following relation \cite{fidelity}:\begin{eqnarray}
f\left(\rho\right) & = & \frac{2F\left(\rho\right)+1}{3}.\label{Tele-fidelity}\end{eqnarray}
Let us note that without shared entanglement the best possible fidelity
for teleportation (classical fidelity) of a completely unknown qubit
is given by $2/3$ \cite{henderson-2000}. Therefore, to outperform
a classical strategy with shared entanglement $\rho$, the condition
$F\left(\rho\right)>1/2$ must be satisfied. In the context of entanglement
distillation the same condition, i.e., $F\left(\rho\right)>1/2$,
determines whether $\rho$ can be distilled by the existing distillation
protocols \cite{distillation-1,distillation-3,Distillation04}. 

Typically, questions related to entanglement distillation and fidelity
pre-suppose that Alice and Bob already share a single copy of a mixed
entangled state $\rho$ or many copies of it. In this work we take
a step backward and ask the following: Given a quantum channel $\Lambda$,
what is the maximum achievable fidelity and what is the best strategy
to establish an entangled state for which this optimal fidelity is
attained? We consider these questions when only a single use of the
channel and local post-processing by trace-preserving operations are
allowed. The first condition implies that we are only interested in
establishing a \emph{single} copy of an entangled state, and the second
condition ensures that there is no particle loss under local operations.
The purpose of this paper is to explicitly demonstrate the counter-intuitive
nature of the answers that may be obtained in this setting.

Before we get to our results it is necessary to recall some very useful
results on fidelity. For separable states it is known that $F=1/2$.
Surprisingly there exist entangled states for which $F\leq1/2$ \cite{Badziag,Bandyopadhyay,F<=1/2},
implying that such states are not directly useful for teleportation.
Nevertheless, by local filtering, fidelity of such entangled states
can be brought above $1/2$ so that they become useful for both teleportation
and distillation \cite{increaseF}. Local filtering \cite{filtering-1,filtering02},
however, is not trace-preserving: It succeeds only with some non-zero
probability and in case of a failure the state becomes separable.
Interestingly, in Refs.$\,$\cite{Badziag,Bandyopadhyay} examples
of mixed entangled states with $F\leq1/2$ were given whose fidelity
can be increased beyond $1/2$ by trace-preserving local operations
and classical communication (TP LOCC). Subsequently, it was proved
that a state of two qubits is entangled if and only if under TP LOCC
its fidelity exceeds $1/2$ \cite{VV-I} . This led the authors in
Ref.$\,$\cite{VV-I} to define the maximum achievable fidelity $F^{*}\left(\rho\right)$
for any $2\otimes2$ density matrix $\mbox{\ensuremath{\rho}}$ as,\begin{eqnarray}
F^{*}\left(\rho\right) & = & \max_{\mbox{TP LOCC}}F\left(\rho\right)\geq F\left(\rho\right).\label{F-star-rho}\end{eqnarray}
While the exact analytical expression $F^{*}\left(\rho\right)$ is
not known, it can be obtained by solving a convex semidefinite program
\cite{VV-I}. Moreover, $F^{*}\left(\rho\right)$ was shown to be
an entanglement monotone; in particular, it quantifies the minimal
amount of mixing required to destroy the entanglement of $\rho$ \cite{VV-I}.
Here one should note that fidelity $F\left(\rho\right)$ is not an
entanglement monotone as it can increase under TP LOCC.

\emph{Formulation of the problem.} To answer the questions raised
in the beginning of the paper it is necessary to consider a two-step
process. In the first step, Alice prepares a two qubit pure entangled
state, say, $|\chi\rangle$ and sends the second qubit through the
quantum channel $\Lambda$. This results in a mixed state, possibly
entangled, $\rho\left(\chi,\Lambda\right)$ shared between them. As
it is possible to enhance the fidelity $F\left(\rho\left(\chi,\Lambda\right)\right)$
of this state by TP LOCC, the second step constitutes Alice and Bob
performing optimal trace-preserving local operations to attain the
maximum fidelity. Let us therefore define the quantity of interest:\begin{eqnarray}
\mathcal{F}\left(\Lambda\right) & = & \max_{|\chi\rangle}F^{*}\left(\rho\left(\chi,\Lambda\right)\right).\label{F-Lambda}\end{eqnarray}
We call the quantity $\mathcal{F}\left(\Lambda\right)$ maximum achievable
fidelity or optimal fidelity for the channel $\Lambda$. Clearly,
given a quantum channel $\Lambda$, the objective of Alice and Bob
is to maximize $F\left(\rho\left(\chi,\Lambda\right)\right)$ over
all TP LOCC and $|\chi\rangle$. We note that it is important to distinguish
$\mathcal{F}\left(\Lambda\right)$ from the channel fidelity considered
in Ref.$\,$\cite{fidelity}. From Eq.$\,$(\ref{Tele-fidelity})
we can also obtain the optimal teleportation fidelity for a single
use of the channel $\Lambda$: \begin{eqnarray*}
f\left(\Lambda\right) & = & \frac{2\mathcal{F}\left(\Lambda\right)+1}{3}.\end{eqnarray*}
\emph{Amplitude damping channel.} The quantum channel considered in
this work is the amplitude damping channel. The action of an amplitude
damping channel $\Lambda$ on a qubit $\sigma$ is given by:\begin{eqnarray}
\sigma\rightarrow\Lambda\left(\sigma\right) & = & M_{0}\sigma M_{0}^{\dagger}+M_{1}\sigma M_{1}^{\dagger},\label{amp-damp}\end{eqnarray}
where $M_{0}$ and $M_{1}$ are the Krauss operators defined by \begin{equation}
M_{0}=\left[\begin{array}{cc}
1 & 0\\
0 & \sqrt{1-p}\end{array}\right],\;\; M_{1}=\left[\begin{array}{cc}
0 & \sqrt{p}\\
0 & 0\end{array}\right],\label{amp-kraus}\end{equation}
with the real parameter $0\leq p\leq1$ characterizing the strength
of the channel. The channel is trace preserving, that is, $\sum_{i=0,1}M_{i}^{\dagger}M_{i}=I$.
For the noise-free case $p=0$, otherwise $0<p\leq1$. For $p=1$
the channel is entanglement breaking \cite{Note for p=1}. Therefore,
throughout this paper we only consider values of $0<p<1$. We note
that $\mathcal{F}\left(\Lambda\right)$ is a function of $p$ alone.

\emph{Summary of the results.} Intuition suggests that for any channel
$\Lambda$ the best strategy to obtain optimal fidelity is to send
part of a maximally entangled state across the channel plus local
post-processing i.e., the relation \begin{eqnarray}
\mathcal{F}\left(\Lambda\right) & = & F^{*}\left(\rho\left(\Phi^{+},\Lambda\right)\right),\label{fundu-relation}\end{eqnarray}
 where $|\Phi^{+}\rangle=\frac{1}{\sqrt{2}}\left(|00\rangle+|11\rangle\right)$
should be true. But as will be demonstrated here, the above relation
does \emph{not} hold in general. We show that the maximum achievable
fidelity $\mathcal{F}\left(\Lambda\right)$ is attained for non-maximally
entangled states for all $p$, $0<p<1$; i.e., \begin{eqnarray}
\mathcal{F}\left(\Lambda\right) & =F^{*}\left(\rho\left(\chi_{0},\Lambda\right)\right) & >F^{*}\left(\rho\left(\Phi^{+},\Lambda\right)\right),\label{fundu-1-a}\end{eqnarray}
where $|\chi_{0}\rangle$ is a non-maximally entangled state. And
when the channel is very noisy, that is, $p\approx1$, the entanglement
of $|\chi_{0}\rangle$ becomes vanishingly small, and yet it gives
the optimal value for fidelity over all transmitted states, including
maximally entangled under trace-preserving local operations. Surprisingly,
we find that to achieve the optimal value, local post-processing is
not be required: i.e., \begin{equation}
\mathcal{F}\left(\Lambda\right)=F^{*}\left(\rho\left(\chi_{0},\Lambda\right)\right)=F\left(\rho\left(\chi_{0},\Lambda\right)\right).\label{fundu-II}\end{equation}
Thus the pre-processed fidelity obtained simply by sending one half
of the \emph{appropriate} non-maximally entangled state through the
channel is actually optimal.

A consequence of the first result is that the ordering of entangled
states under some entanglement monotone can be reversed even though
the states undergo identical local interaction via a trace-preserving
completely positive map. The argument goes as this. Before the second
qubit underwent interaction with the channel $\Lambda$, we had trivially
$F^{*}\left(\Phi^{+}\right)\geq F^{*}\left(\chi_{0}\right)$. Now
after the interaction our first result implies that\begin{eqnarray}
F^{*}\left(\rho\left(\chi_{0},\Lambda\right)\right) & > & F^{*}\left(\rho\left(\Phi^{+},\Lambda\right)\right).\label{F*>F*phi}\end{eqnarray}
The conclusion now follows by noting that $F^{*}$ is an entanglement
monotone. It is interesting that the ordering does not change for
any pair of transmitted states under concurrence. For example, we
find that $C\left(\rho\left(\Phi^{+},\Lambda\right)\right)>C\left(\rho\left(\chi_{0},\Lambda\right)\right)$
where $C$ is the concurrence \cite{Concurrence}.

We further show that local trace preserving operations can enhance
the fidelity of the states $\rho\left(\chi,\Lambda\right)$ if and
only if $p_{0}<p<1$ and $C\left(\chi\left(q\right)\right)<C\left(\chi\right)\leq1$,
where $q$ is a function of $p$. The first condition implies that
if $p\leq p_{0}$ then $F\left(\rho\left(\chi,\Lambda\right)\right)$
cannot be increased by TP LOCC for any $|\chi\rangle$. The second
condition on the other hand shows that when $p>p_{0}$, fidelity can
be increased only for a subset of states $\rho\left(\chi,\Lambda\right)$:
in particular those resulting from the transmission of states $|\chi\rangle$
with relatively higher entanglement.

\textbf{Remark:} In the above results both $\mathcal{F}\left(\Lambda\right)$
and $|\chi_{\mbox{0}}\rangle$ are functions of the channel parameter
$p$. This means that for different values of $p$ different optimal
values of fidelity are obtained. The corresponding non-maximally entangled
states are different as well.

\emph{Details of the results.} We shall now prove the results. Alice
prepares a two qubit pure entangled state $|\chi\rangle=\alpha|00\rangle+\beta|11\rangle$,
where $\alpha,\beta$ are real and satisfy the conditions $\alpha\geq\beta>0$
and $\alpha^{2}+\beta^{2}=1$. She sends the second qubit through
the amplitude damping channel defined by Eq.$\,$(\ref{amp-damp}).
We therefore have,\begin{eqnarray}
\rho\left(\chi\right)\rightarrow\rho\left(\chi,\Lambda\right) & = & \sum_{i=0,1}\left(I\otimes M_{i}\right)\rho\left(\chi\right)\left(I\otimes M_{i}^{\dagger}\right),\label{rho2rhoLambda}\end{eqnarray}
where $\rho\left(\chi\right)=|\chi\rangle\langle\chi|$. The final
state $\rho\left(\chi,\Lambda\right)$ can be conveniently expressed
as: \begin{eqnarray*}
\rho\left(\chi,\Lambda\right) & = & \left(1-p\beta^{2}\right)|\eta\rangle\langle\eta|+p\beta^{2}|01\rangle\langle01|,\end{eqnarray*}
where \begin{eqnarray*}
|\eta\rangle & = & \frac{\alpha}{\sqrt{1-p\beta^{2}}}|00\rangle+\frac{\beta\sqrt{1-p}}{\sqrt{1-p\beta^{2}}}|11\rangle\end{eqnarray*}
We first obtain the fidelity $F\left(\rho\left(\chi,\Lambda\right)\right)$
before any postprocessing is performed. Define a real $3\times3$
matrix $T$ whose elements are given by $t_{ij}=\mbox{Tr}\left[\rho\left(\chi,\Lambda\right)\sigma_{i}\otimes\sigma_{j}\right]$,
where $\sigma_{i}$s are the Pauli matrices. In our case $T$ is diagonal
and $\mbox{det}T$ is negative. For the states with diagonal $T$
and $\det T<0$, $F$ is given by \cite{Badziag},\begin{eqnarray}
F & = & \frac{1}{4}\left(1+\sum_{i}|t_{ii}|\right),\label{exp-fidelity}\end{eqnarray}
which in our case turns out to be,\begin{eqnarray}
F\left(\rho\left(\chi,\Lambda\right)\right) & = & \frac{1}{2}\left(1+2\alpha\beta\sqrt{1-p}-p\beta^{2}\right).\label{chi-fidelity}\end{eqnarray}
The concurrence \cite{Concurrence} of $\rho\left(\chi,\Lambda\right)$
is given by $C=2\alpha\beta\sqrt{1-p}$. It is easy to check that
$F$ is not always greater then $1/2$ even though $C\left(\rho\left(\chi,\Lambda\right)\right)$
is always non zero as long as $p\neq1$. For example, if $|\chi\rangle=|\Phi^{+}\rangle$,
then for all values $p\geq2\left(\sqrt{2}-1\right)$, $F\leq1/2$. 

The useful observation to be made here is that the maximum of $F$
(for any $p$, $0<p<1$) is not obtained when $|\chi\rangle=|\Phi^{+}\rangle.$
In particular,\begin{eqnarray}
F_{\max} & = & F\left(\rho\left(\chi_{0},\Lambda\right)\right)=1-\frac{p}{2},\label{Fmax}\end{eqnarray}
where\begin{eqnarray}
|\chi_{0}\rangle & = & \frac{1}{\sqrt{2-p}}|00\rangle+\sqrt{\frac{1-p}{2-p}}|11\rangle.\label{chi-zero}\end{eqnarray}
It is worth noting that $F_{\mbox{max}}$ is the maximum eigenvalue
of the density matrix $\rho\left(\Phi^{+},\Lambda\right)$, and $|\chi_{0}\rangle$
is the corresponding eigenstate. Indeed, for any quantum channel $\$$,
the maximum pre-processed fidelity is given by the maximum eigenvalue
of the density matrix $\rho\left(\Phi^{+},\$\right)$ and is obtained
by sending one half of the corresponding eigenstate through the channel
(see Ref. {[}21{]} for details). 

Equation$\,\left(\ref{Fmax}\right)$ while surprising, is not conclusive
because the maximum achievable fidelity $\mathcal{F}$ may still be
obtained for $|\chi\rangle=|\Phi^{+}\rangle$ \emph{after} Alice and
Bob perform trace-preserving LOCC: i.e., the possibility of $\mathcal{F}\left(\Lambda\right)=F^{*}\left(\rho\left(\Phi^{+},\Lambda\right)\right)$
cannot be ruled out immediately. The following proposition, however,
negates this possibility.

\begin{prop} $\mathcal{F}\left(\Lambda\right)>F^{*}\left(\rho\left(\Phi^{+},\Lambda\right)\right)$
for any $p,$ where $0<p<1$. \end{prop} 

\begin{proof} The result can be proved by computing $F^{*}\left(\rho\left(\Phi^{+},\Lambda\right)\right)$
{[}see Eqs.$\,$(\ref{Fstar-chi-1A}) and (\ref{Fstar-chi-1B}){]}.
Here we give an alternative proof which does not require computing
it explicitly. We first note that by definition $\mathcal{F}\left(\Lambda\right)\geq F_{\mbox{max}}$,
where $F_{\mbox{max}}$ is given by (\ref{Fmax}). Now, for any density
matrix $\rho$, $F^{*}\left(\rho\right)\leq\frac{1}{2}\left(1+N\left(\rho\right)\right),$where
$N\left(\rho\right)=\max\left[0,-2\lambda_{\mbox{min}}\left(\rho^{\Gamma}\right)\right]$
and $\rho^{\Gamma}$ is partial transpose of $\rho$ \cite{F<=1/2}.
Importantly, the equality is achieved iff the eigenvector corresponding
to the negative eigenvalue of $\rho^{\Gamma}$ is maximally entangled
\cite{F<=1/2}. It can be easily checked that the eigenvector corresponding
to the negative eigenvalue of $\rho^{\Gamma}\left(\Phi^{+},\Lambda\right)$
is not maximally entangled unless $p=0$. It therefore follows that\begin{eqnarray}
F^{*}\left(\rho\left(\Phi^{+},\Lambda\right)\right) & < & \frac{1}{2}\left[1+N\left(\rho\left(\Phi^{+},\Lambda\right)\right)\right]\nonumber \\
 & = & 1-\frac{p}{2}=F_{\mbox{max}}\leq\mathcal{F}\left(\Lambda\right)\label{prop-1-eq-2}\end{eqnarray}
This concludes the proof. \end{proof} 

\textbf{Remark:} As we have explained before, the above result shows
that a trace-preserving completely positive map can reverse the ordering
of entangled states for the entanglement monotone $F^{*}$. Here we
simply note that this reversal is not present when the entanglement
measure is concurrence. It is easy to see that for any pair of pure
states $|\chi_{1}\rangle,|\chi_{2}\rangle$, if $C\left(\chi_{1}\right)\geq C\left(\chi_{2}\right)$
then after the interaction $C\left(\rho\left(\chi_{1},\Lambda\right)\right)\geq C\left(\rho\left(\chi_{2},\Lambda\right)\right)$,
where $C\left(\rho\left(\chi,\Lambda\right)\right)=2\alpha\beta\sqrt{1-p}$. 

We will now obtain an exact expression for $F^{*}\left(\rho\left(\chi,\Lambda\right)\right)$
for any $|\chi\rangle$. In Ref.$\,$\cite{VV-I} it was shown that
for any given $2\otimes2$ density matrix $\rho$ the maximum achievable
fidelity $F^{*}\left(\rho\right)$ by TP LOCC can be found by solving
the convex semidefinite program: Maximize\begin{eqnarray}
F^{*} & = & \frac{1}{2}-\mbox{Tr}\left(X\rho^{\Gamma}\right)\label{Fstar-maximize}\end{eqnarray}
under the constraints,\begin{eqnarray*}
0\leq & X & \leq I_{4}\\
-\frac{I_{4}}{2}\leq & X^{\Gamma} & \leq\frac{I_{4}}{2}\end{eqnarray*}
where $X$ is a $4\times4$ matrix and $\Gamma$ denotes partial transposition.
Moreover, the optimal $X$ is of rank $1$. Solving the above in our
case using the symmetries of the state $\rho\left(\chi,\Lambda\right)$,
we obtain the following expressions for maximum achievable fidelity:\begin{eqnarray}
F^{*}\left(\rho\left(\chi,\Lambda\right)\right)=F_{1}^{*} & = & \frac{1}{2}\left(1+2\alpha\beta\sqrt{1-p}-p\beta^{2}\right)\label{Fstar-chi-1A}\\
 & \mbox{if} & \frac{p^{2}}{1-p+p^{2}}\leq\alpha^{2}<1\nonumber \\
F^{*}\left(\rho\left(\chi,\Lambda\right)\right)=F_{2}^{*} & = & \frac{1}{2}\left(1+\alpha^{2}\frac{1-p}{p}\right),\label{Fstar-chi-1B}\\
 & \mbox{if} & \frac{1}{2}\leq\alpha^{2}<\frac{p^{2}}{1-p+p^{2}}\nonumber \end{eqnarray}
Maximum achievable fidelity for any ordered pair $\left(p,|\chi\rangle\right)$
can be obtained from the above equations. Let $g\left(p\right)=\frac{p^{2}}{1-p+p^{2}}$.
We first observe that the cases corresponding to $F_{2}^{*}$ arise
only when $g\left(p\right)>\frac{1}{2},$ or equivalently $p>\frac{1}{2}\left(\sqrt{5}-1\right)=p_{0}$.
Therefore, when $p\leq p_{0}$, then for any state $|\chi\rangle$,
we have $F^{*}=F_{1}^{*}=F$, where the last equality follows by comparing
Eqs.$\,$(\ref{chi-fidelity}) and (\ref{Fstar-chi-1A}). In these
cases, therefore, there is no benefit from local processing of the
states $\rho\left(\chi,\Lambda\right)$. On the other hand, when $p>p_{0}$,
the question of enhancing the fidelity of $\rho\left(\chi,\Lambda\right)$
depends on entanglement of the state $|\chi\rangle$. For any $p$,
where $p_{0}<p<1$, the transmitted states $|\chi\rangle$ fall in
two distinct classes: (a) those satisfying $\frac{1}{2}\leq\alpha^{2}<g\left(p\right)$
or equivalently $C\left(g\left(p\right)\right)<C\left(\chi\right)\leq1$,
and (b) those for which $g\left(p\right)\leq\alpha^{2}<1$ or equivalently
$0<C\left(\chi\right)\leq C\left(g\left(p\right)\right)$. It is to
be understood that $C\left(g\left(p\right)\right)$ is the shorthand
notation of the concurrence of the state $\left|\chi\left(\alpha^{2}=g\left(p\right)\right)\right\rangle $.
Now every state in class (a) is \emph{more} entangled than every state
in class (b). Therefore, when $p>p_{0}$, the fidelity of the resulting
mixed states can only be increased if the transmitted state belongs
to class (a), that is, the class of states with relatively higher
entanglement. Summarizing the above we have the next proposition. 

\begin{prop} Local trace preserving operations can enhance the fidelity
of the states $\rho\left(\chi,\Lambda\right)$ if and only if $p_{0}<p<1$
and $C\left(\chi\left(q\right)\right)<C\left(\chi\right)\leq1$, where
$q=g\left(p\right)$. \end{prop}

Equations (\ref{Fstar-chi-1A}) and (\ref{Fstar-chi-1B}) contain
all information that we need to know to obtain $\mathcal{F}\mbox{\ensuremath{\left(\Lambda\right)}}$.
Let us denote\begin{eqnarray*}
\mathbb{F}_{1}\left(\Lambda\right) & = & \max_{|\chi\rangle}F_{1}^{*},\end{eqnarray*}
where the maximum is taken over all pure states $|\chi\rangle$ satisfying
the condition $g\left(p\right)\leq\alpha^{2}<1$, and \begin{eqnarray*}
\mathbb{F}_{2}\left(\Lambda\right) & = & \max_{|\chi\rangle}F_{2}^{*}\end{eqnarray*}
where the maximum is taken over all pure states $|\chi\rangle$ satisfying
the condition $\frac{1}{2}\leq\alpha^{2}<g\left(p\right)$. Thus the
optimal fidelity for the channel is given by \begin{eqnarray}
\mathcal{F}\left(\Lambda\right) & = & \mathbb{F}_{1}\left(\Lambda\right)\;\mbox{\;\;\;\;\;\;\;\;\;\;\;\;\,\:\;\;\;\;\;\;\;\;\; if}\; p\leq p_{0},\label{F1-Lambda}\\
\mathcal{F}\left(\Lambda\right) & = & \max\left\{ \mathbb{F}_{1}\left(\Lambda\right),\mathbb{F}_{2}\left(\Lambda\right)\right\} \;\;\mbox{if}\; p>p_{0},\label{F2-Lambda}\end{eqnarray}
where $p_{0}=\frac{1}{2}\left(\sqrt{5}-1\right)$.

\begin{prop} The maximum achievable fidelity $\mathcal{F}\left(\Lambda\right)$
is given by $F_{\max}=1-\frac{p}{2}$ for all $p$, $0<p<1$.\end{prop}

\begin{proof} From Eqs. (\ref{F1-Lambda}) and (\ref{F2-Lambda})
it is clear that two cases have to be considered. We first consider
the case when $p\leq p_{0}$. First observe that $F_{1}^{*}=F\left(\rho\left(\chi,\Lambda\right)\right)$.
Therefore,\begin{eqnarray*}
\mathbb{F}_{1}\left(\Lambda\right)=\max_{|\chi\rangle}F_{1}^{*} & = & \max_{|\chi\rangle}F\left(\rho\left(\chi,\Lambda\right)\right),\end{eqnarray*}
where the maximum is taken over all pure states $|\chi\rangle$ such
that $\alpha^{2}\geq g\left(p\right)$. But $F_{\mbox{max}}$ is obtained
for the state $|\chi_{0}\rangle$ given by Eq.$\,$(\ref{chi-zero})
which already satisfies the condition $\alpha^{2}=\frac{1}{2-p}>g\left(p\right)$
for any $p$, $0<p<1$. Thus we have proven that, for $p\leq p_{0}$
\[
\mathcal{F}\left(\Lambda\right)=F_{\mbox{max}}=1-\frac{p}{2},\]
We now consider the case when $p>p_{0}$. From Eq.$\,$(\ref{Fstar-chi-1B})
we can get an upper bound on $\mathbb{F}_{2}\left(\Lambda\right)$,
\begin{eqnarray*}
\mathbb{F}_{2}\left(\Lambda\right) & < & \frac{1}{2}\left(1+g\left(p\right)\frac{1-p}{p}\right)\end{eqnarray*}

It is now easy to check that $F_{\mbox{max}}>\frac{1}{2}\left(1+g\left(p\right)\frac{1-p}{p}\right)$
for every $p$, $0<p<1$. Thus $F_{\mbox{max}}>\mathbb{F}_{2}\left(\Lambda\right)$.
This implies that if $p>p_{0}$, the optimal fidelity is not attained
by any pure state satisfying $\frac{1}{2}\leq\alpha^{2}<g\left(p\right)$.
Instead the optimal fidelity is obtained, once again, for the state
$|\chi_{0}\rangle$. Noting that $\mathbb{F}_{1}=F_{\mbox{max}}$,
we have therefore proven that for $p>p_{0}$\[
\mathcal{F}\left(\Lambda\right)=F_{\mbox{max}}=1-\frac{p}{2}.\]
This concludes the proof. \end{proof} 

\textbf{Remark:} The maximum achievable fidelity $\mathcal{F}\left(\Lambda\right)$
being equal to $F_{\mbox{max}}$ shows that post-processing by TP
LOCC is not necessary to achieve the optimal value as long as the
appropriate non-maximally entangled state $|\chi_{0}\left(p\right)\rangle$
is transmitted. This also suggests that enhancing of fidelity by TP
LOCC is possibly a sub-optimal phenomenon. While TP LOCC can certainly
increase fidelity for some states, it may not be the case that the
optimal fidelity for the channel is obtained that way.

\textbf{Remark:} The concurrence of $|\chi_{0}\rangle$ for which
the optimal fidelity is obtained is given by $C\left(\chi_{0}\right)=2\sqrt{1-p}/\left(2-p\right)$.
Because $C\left(\chi_{0}\right)$ is a monotonically decreasing function
of $p$, this shows that if the channel is very noisy, that is, $p\approx1$,
the concurrence of the state $|\chi_{0}\rangle$ becomes arbitrarily
close to zero. Perhaps more interesting is the behavior of $C\left(\chi_{0}\right)$
with $p$. Figure 1 shows that the concurrence decreases with $p$
rather slowly until $p$ enters the {}``very noisy'' domain, wherein
it starts to fall quite rapidly. For example, for $p=0.75$, $C\left(\chi_{0}\right)=0.8$,
whereas for $p=0.999$, $C\left(\chi_{0}\right)=0.063$.

\emph{Discussions:} Several interesting questions arise in the context
of the results reported. For example, for which other quantum channels
can similar results be observed? A possible way to explore this is
to characterize the quantum channels where the maximum fidelity (before
any post processing by TP LOCC) is obtained by non-maximally entangled
states. The channels that show this behavior are those with the property
that the eigenvector corresponding to the maximum eigenvalue of $\rho\left(\Phi^{+},\$\right)$
($\$$ is a quantum channel) is not maximally entangled \cite{verstaete-2002}.
The amplitude damping channel belongs to this class but phase damping
and depolarizing channels do not. Thus if the channel is phase damping
or depolarizing then the maximum pre-processed fidelity is always
attained by sending part of a maximally entangled state through the
channel. Despite these observations, a complete characterization of
channels exhibiting properties such as the ones presented here should
be an interesting problem for future studies. Another question of
interest is whether similar results can be observed in the regime
of finite copies. The regime of finite copies is non-asymptotic but
of considerable practical interest because such cases may be realized
experimentally. 

\emph{Conclusions:} To conclude, we have investigated the question
of optimal fidelity for a given quantum channel and what is the best
protocol to achieve the optimal value. While the results presented
in this paper illustrate many interesting features that go against
conventional intuition, it is likely that they are not generic features
of quantum channels. Nevertheless we certainly hope that they would
contribute to our understanding of quantum channels and fidelity. 

\includegraphics[clip,width=10cm,height=9cm]{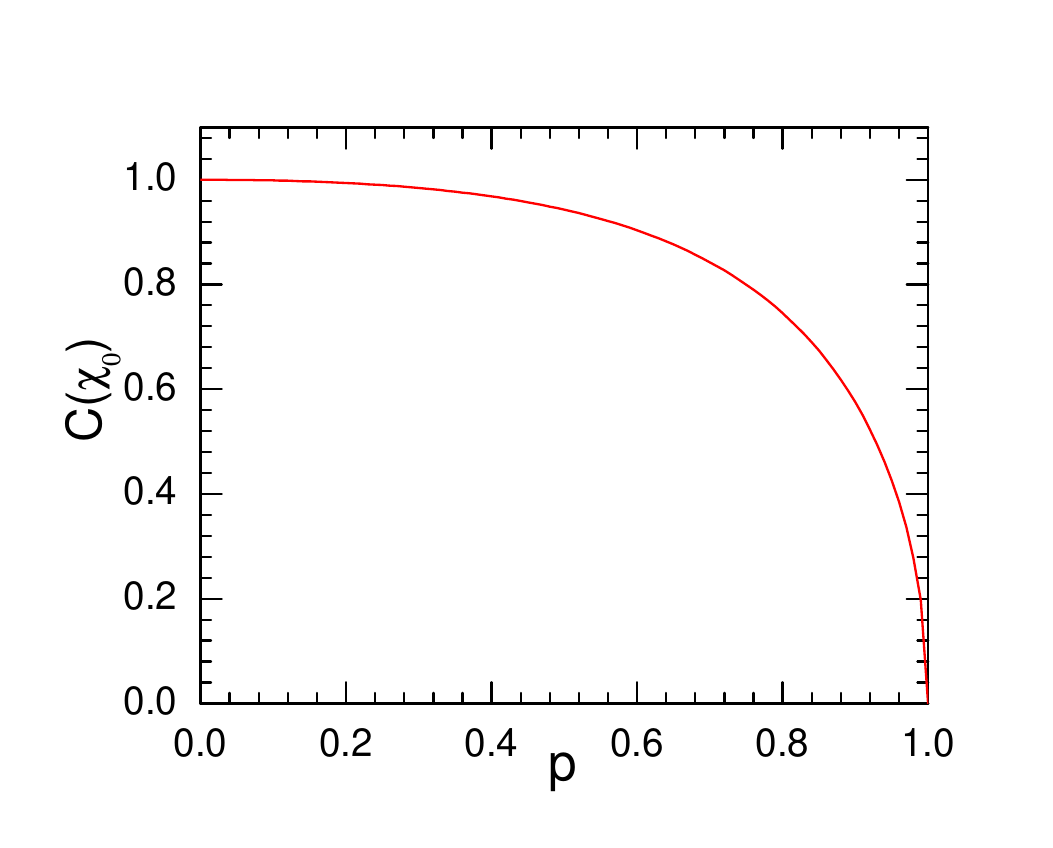}

Fig. 1: Concurrence of $|\chi_{0}\rangle$ vs channel parameter $p$. 

\begin{thebibliography}{21}
\bibitem{Note for p=1}For $p=1$, transmission of part of any two
qubit pure state through the channel makes it separable.
\end{thebibliography}
\end{document}